\documentclass[11pt]{amsart}
\usepackage{amsmath,amssymb,amsfonts,amsbsy}
\usepackage{graphicx}
\usepackage{float}
\usepackage{setspace}
\usepackage[margin=1.0in]{geometry}

\newtheorem{thm}{Theorem}[section]
\newtheorem{lem}[thm]{Lemma}
\newtheorem{cor}[thm]{Corollary}
\newtheorem{prop}[thm]{Proposition}

\theoremstyle{definition}

\theoremstyle{definition}

\theoremstyle{definition}
\newtheorem{remark}[thm]{Remark}

\newcommand{\mc}[1]{\mathcal{#1}}
\newcommand{\e}[1]{\emph{#1}}

\newcommand{\la}{\langle}
\newcommand{\ra}{\rangle}
\newcommand{\tr}{\mathrm{tr}}

\newcommand{\rmv}[1]{}

\newcommand{\hs}{\hskip5pt}

\newcommand{\field}[1]{\mathbb{#1}}
\newcommand{\C}{{\field{C}}}

\newcommand{\id}{{\iota}}

\newcommand{\om}{{\omega}}

\newcommand{\bd}{{\mathcal{B}}}
\newcommand{\N}{\mathcal{N}}
\newcommand{\M}{{\mathfrak{M}}}

\newcommand{\B}{{\mathcal{B}}}

\newcommand{\fee}{{\varphi}}

\newcommand{\LL}{{L^{\infty}(\G)}}
\newcommand{\Ll}{{\ell^{\infty}(\G)}}
\newcommand{\Lo}{{\ell^{1}(\G)}}

\newcommand{\LLH}{{L^{\infty}(\hat{\G})}}

\newcommand{\LO}{L^{1}(\G)}
\newcommand{\LT}{L^{2}(\G)}

\newcommand{\Lp}{{\mathcal{L}^p(\G)}}

\newcommand{\LLL}{L^{\infty}(\hat\G)}

\newcommand{\R}{{\field{R}}}

\newcommand{\abs}[1]{|#1|}
\newcommand{\norm}[1]{\|#1\|}

\newcommand{\G}{\mathbb G}

\newcommand{\RG}{\mc{R}(G)}

\newcommand{\Nphi}{\mc{N}_\varphi}
\newcommand{\Mphi}{\mc{M}_\varphi}

\newcommand{\vphi}{\varphi}

\newcommand{\lm}{\lambda}
\newcommand{\Lm}{\Lambda}
\newcommand{\Gam}{\Gamma}

\newcommand{\ten}{\otimes}

\newcommand{\h}[1]{\hat{#1}}

\providecommand{\abs}[1]{\lvert#1\rvert}
\providecommand{\norm}[1]{\lVert#1\rVert}
\providecommand{\bignorm}[1]{\bigg\lVert#1\bigg\rVert}

\newcommand{\F}{\mc{F}}

\begin{document}

\title{An uncertainty principle for unimodular quantum groups}
\author{Jason Crann$^{1,2}$ and Mehrdad Kalantar$^{1}$}
\email{jason\_crann@carleton.ca, mkalanta@math.carleton.ca}
\address{$^1$School of Mathematics \& Statistics, Carleton University, Ottawa, ON, Canada K1S 5B6}
\address{$^2$Universit\'{e} Lille 1 - Sciences et Technologies, UFR de Math\'{e}matiques, Laboratoire de Math\'{e}matiques Paul Painlev\'{e}
- UMR CNRS 8524, 59655 Villeneuve d'Ascq C\'{e}dex, France}

\begin{abstract} We present a generalization of Hirschman's entropic uncertainty principle for locally compact abelian groups to unimodular
locally compact quantum groups. As a corollary, we strengthen a well-known uncertainty principle for compact groups, and generalize the relation to compact
quantum groups of Kac type. We also establish the complementarity of finite-dimensional quantum group algebras. In the non-unimodular setting,
we obtain an uncertainty relation for arbitrary locally compact groups using the relative entropy with respect to the Haar weight as the measure of uncertainty.
We also show that when restricted to $q$-traces of discrete quantum groups, the relative entropy with respect to the Haar weight reduces to the
canonical entropy of the random walk generated by the state.\end{abstract}
% - a result of independent interest which may lead to applications in non-commutative
%random walks on discrete quantum groups.\end{abstract}

\maketitle

\begin{spacing}{1.3}

\section{Introduction}

Heisenberg's celebrated uncertainty principle \let\thefootnote\relax\footnotetext{2010 \e{Mathematics Subject Classification} Primary: 46L89, 81R15; Secondary: 22D25, 81R05. The first author was supported by an NSERC Canada Graduate Scholarship.}asserts the mutual incompatibility of measurements of position and momentum on $L^2(\R)$, in the
sense that the product of their uncertainties in any state is bounded below by some universal constant \cite{Hei}. This was later quantified by Kennard \cite{K},
who showed that
\begin{equation}\label{Kenn}\sigma(Q,f)\sigma(P,f)\geq\frac{\hbar}{2}\end{equation}
for any $\norm{f}_2=1$, where $\sigma(Q,f)$ is the standard deviation of a measurement of $Q$ in the state $f$. As $Q$ and $P$
are unitarily equivalent via the Fourier transform, one may interpret this uncertainty principle as a statement about the complementarity of a function
$f\in L^2(\R)$ and its Fourier transform $\hat{f}\in L^2(\R)$. Indeed, it was shown by Hirschman \cite{Hir} that
\begin{equation}\label{Hirschmann}H(|f|^2)+H(|\hat{f}|^2)\geq0\end{equation}
for all $\norm{f}_2=1$, where $H(|f|^2)$ is the entropy of the density $|f|^2\in L^1(\R)$. This was later sharpened by Beckner \cite{Bec} to
\begin{equation*}H(|f|^2)+H(|\hat{f}|^2)\geq\log(\pi e).\end{equation*}
Under the convention that $\hbar\equiv1$, this latter inequality implies (\ref{Kenn}),
suggesting that entropy may be more suitable for measuring the complementarity of $f$ and $\hat{f}$. Moreover, Hirschman remarks that a similar argument as in \cite{Hir} yields inequality (\ref{Hirschmann}) for arbitrary locally compact abelian groups.

%showed that for any locally compact abelian group $G$, and any $f\in L^2(G)$ with $\norm{f}_2=1$,
%\begin{equation}\label{Hirsch}H(|f|^2)+ H(|\hat{f}|^2)\geq0,\end{equation}
%where $H(|f|^2)$ and $ H(|\hat{f}|^2)$ are the entropies of $f\in L^2(G)$ and $\hat{f}\in L^2(\hat{G})$ with respect to fixed
%left Haar measures on $G$ and $\hat{G}$, respectively.

With non-abelian group duality fully established, along with the entropy theory of normal states on von Neumann algebras, a natural question
is to seek a manifestation of Hirschman's entropic uncertainty principle in this more general setting.
In this paper, we present such a generalization to the level of unimodular locally compact quantum groups.
As a corollary, we strengthen a well-known uncertainty principle for compact groups, and generalize the relation to compact
quantum groups of Kac type. We also show that the algebras $\LL$ and $\LLH$ associated to a finite-dimensional quantum group
$\G$ and its dual $\hat{\G}$ are complementary in the sense of Petz \cite{Petz}, and satisfy a non-commutative analog of the well-known
uncertainty relation for mutually unbiased bases.

Towards the non-unimodular generalization, in the final section we establish an entropic uncertainty principle for arbitrary locally compact groups
by using the relative entropy with respect to the Haar weight as the measure of uncertainty. As a side result, we also show that when restricted to $q$-traces of
discrete quantum groups, the relative entropy with respect to the Haar weight reduces to the
canonical entropy of the random walk generated by the state (cf. \cite[\S2]{HI}).

We begin with a brief overview of the relevant tools from locally compact quantum groups. For more details on the subject we refer the reader to \cite{KV}.

A \emph{locally compact quantum group} $\G$ is a
quadruple $(M, \Gamma, \varphi, \psi)$, where $M$ is a
von Neumann algebra with a co-associative co-multiplication
$\Gamma: M\to M \bar\otimes M$, and $\varphi$
and  $\psi$ are  (normal faithful semi-finite) left and right
Haar weights on $M$, respectively.
We write $\Mphi^+ = \{x\in M^+ : \fee(x)<\infty\}$ and
$\Nphi = \{x\in M^+ : \fee(x^*x)<\infty\}$, and we denote
by $\Lambda_\fee$ the inclusion of $\Nphi$ into the GNS Hilbert
space $H_\fee$ of $\fee$.
For each locally compact
quantum group $\G$, there exist a  \emph{left fundamental unitary
operator}  $W$ on $H_\fee\otimes H_\fee$
which satisfies  the  pentagonal relation
and such that the co-multiplication $\Gamma$ on $M$ can be expressed as
\begin{equation*}
\Gamma(x) = W^{*}(1\otimes x)W
\quad(x \in M).
\end{equation*}

Let $M_*$ be the predual of $M$.
The \emph{left regular representation} $\lambda : M_* \to
\B(H_\fee)$ is defined by
 \[
\lambda : M_*\ni f   \,\longmapsto\, \lambda(f) = (f\otimes \iota)(W)
\in \B(H_\fee),
 \]
which is an injective map %and completely contractive algebra homomorphism
from $M_*$ into $\B(H_\fee)$.
Then  $\hat M={\{\lambda(f): f\in M_*\}}''$
is the von Neumann algebra associated with the dual quantum group
$\hat \G$.  It follows that $W \in M \bar \otimes \hat M$.
We also define the completely contractive injection
\[
\hat\lambda:  {\hat M}_*\ni\hat f \,\longmapsto\, \hat\lambda(\hat f) =
(\iota \otimes \hat f)(W)\in M.
\]

If $G$ is a locally compact group, then $\G_a=( L^{\infty}(G),\Gam_a,\vphi_a,\psi_a)$ becomes a \e{commutative} quantum group associated with the
commutative von Neumann algebra $ L^{\infty}(G)$, where the co-multiplication is given by $\Gam_a(f)(s,t)=f(st)$, and $\vphi_a$ and $\psi_a$ are integration
with respect to a left and right Haar measure, respectively. The dual quantum group $\h{\G}_a$ of $\G_a$ is given by
$\G_s=(VN(G),\Gam_s,\vphi_s,\psi_s)$, where $VN(G)=\{\lm(g)\mid g\in G\}''$ is the von Neumann algebra generated by the left regular representation of $G$, the
co-multiplication is $\Gam_s(\lm(g))=\lm(g)\ten\lm(g)$, and $\vphi_s=\psi_s$ is Haagerup's Plancherel weight (cf. \cite[\S VII.3]{Tak2}). The duality of
$\G_a$ and $\G_s$ may be seen as a non-abelian generalization of Pontrjagin--van Kampen duality. Indeed, when $G$ is a locally compact abelian group then
$VN(G)\cong L^{\infty}(\hat{G})$, where $\hat{G}$ is the dual group of $G$, i.e., the locally compact abelian group of continuous
characters $\chi:G\rightarrow\mathbb{T}$.

Let $\G$ be a locally compact quantum group such that the left Haar weight $\fee$ on $\G$ is a trace.
For $1\leq p < \infty$, we denote by $L^p(\G)$ the noncommutative $ L^p$-space associated to $\fee$;
this space is obtained by taking the closure
of the span of $\mc{M}_\fee^+$ with the norm $\|x\|_p := \fee(|x|^p)^{\frac1p}$ (see \cite[\S IX.2]{Tak2} for details).
We denote by $\LL$ the von Neumann algebra $M$. Unless otherwise stated, we canonically identify $\LL$ as a von Neumman subalgebra of $\bd(\LT)$ via left multiplication. %Then there are isometric isomorphisms $M_*\cong \LO$
%and $\Lp \cong \Lq^*$, where $1<p,q<\infty$ and $\displaystyle \frac 1p + \frac 1 q = 1$.
The map
\begin{equation}\label{je1}
\Mphi\ni x\,\longmapsto\, \fee_x \in M_*
\end{equation}
extends to an isometric isomorphism between $\LO$ and $M_*$,
where $\langle\fee_x,y\rangle = \fee(xy)$. We say that $\G$ is {\it unimodular} if $\fee = \psi$ is tracial. In this case $\hat \G$ is unimodular too.

%We have $\M_\fee\subseteq\mp$ for all $1<p$.

For a locally compact quantum group $\G$ with tracial left Haar weight $\fee$, and $x\in \LO^+$ with $\|x\|_1 = 1$, we define the {\it entropy} of $x$ by
\[
H(x,\vphi):=-\fee(x \log x)=-\int_0^\infty\lambda\log\lambda d\fee(e_\lambda)
%- \lim_{\varepsilon \rightarrow 0}\,\int_\varepsilon^\infty\, \lambda \,\log \lambda \, de_\lambda\,,
\]
where $\{e_\lambda\}$ are the spectral projections of $x$. For example, if $G$ is a locally compact group with left Haar measure $\mu_G$, and $f\in L^1(G)^+$
with $\norm{f}_1=1$, then
\begin{equation*}H(f,\mu_G)=-\int_Gf(s)\log(f(s))d\mu_G(s),\end{equation*}
the classical entropy of the probability density $f$. For a state $\rho\in\mc{T}(\LT):=\bd(\LT)_*$, we denote the von Neumann entropy of $\rho$ by $H(\rho,\tr)$.

\section{The Uncertainty Principle}

Inspired by a recent argument of Frank and Lieb \cite{FL}, we will use the following two well-known inequalities from
quantum statistical mechanics, this first of which follows from Klein's inequality.

\begin{lem}\label{Gibbs}\textnormal{[Gibbs Variational Principle]} Let $A$ be a self-adjoint operator on a Hilbert space $H$ such that $\tr(e^{-A})<\infty$. Then for any positive $\rho\in\mc{T}(H)$ with $\tr(\rho)=1$, we have
\begin{equation*}\tr(\rho A)+\tr(\rho\log\rho)\geq-\log\tr(e^{-A})\end{equation*}
with equality if and only if $\rho=e^{-A}/\tr(e^{-A})$.\end{lem}

\begin{lem}\cite[Theorem 4]{Ruskai}\label{GT}\textnormal{[Golden--Thompson Inequality]} Let $A$ and $B$ be self-adjoint operators bounded from above, then
\begin{equation*}\tr(e^{A+B})\leq\tr(e^{A/2}e^{B}e^{A/2}).\end{equation*}\end{lem}

\begin{lem}\label{Kraus} Let $\G$ be a unimodular locally compact quantum group and let $\om\in M_*^+$. Then there exists a net $(\hat{w}_k)$ in $\hat{M}$ satisfying $\sum_{k\in K}\hat{w}_k^*\hat{w}_k=\sum_{k\in K}\hat{w}_k\hat{w}_k^*=\om(1)1$, and
\begin{equation*}\Theta(\om)(T):=(\om\ten\id)W^*(1\ten T)W=\sum_{k\in K}\hat{w}_k^*T\hat{w}_k,\hs\hs T\in\bd(\LT),\end{equation*}
where all sums converge in the weak* topology of $\bd( L^2(\G))$.\end{lem}

\begin{proof} Since $\Theta(\om)(T)=(\om\ten\id)W^*(1\ten T)W,\hs T\in\bd(\LT)$, defines a normal completely positive $\hat{M}'$-bimodule map on $\bd(\LT)$, there exists a net $(\hat{w}_k)_{k\in K}$ in $\hat{M}$ satisfying
\begin{equation*}(\om\ten\id)W^*(1\ten T)W=\sum_{k\in K}\hat{w}_k^*T\hat{w}_k\end{equation*}
for all $T\in\bd(\LT)$ \cite{Haa}. Moreover, since $\LL$ is standardly represented on $\LT$, we have $\om=\om_\xi|_{\LL}$ for some vector $\xi\in\LT$. Thus, resolving the identity with any orthonormal basis $(e_k)_{k\in K}$ yields a Kraus decomposition of $\Theta(\om)$ with $\hat{w}_k=(\om_{\xi,e_k}\ten\id)(W)$. Clearly, $\sum_{k\in K}\hat{w}_k^*\hat{w}_k=\om(1)1$. To obtain the remaining sum we exploit unimodularity and use the involution on $M_*$, which yields a new element $\om^o\in M_*^+$ given by $\om^o(x)=\om(\hat{J}x^*\hat{J})$ for $x\in\LL$, where $\hat{J}$ is the conjugate linear isometry arising from the standard representation of $\hat{M}$ on $\LT$. It follows that $\om^o=\om_{\hat{J}\xi}|_{\LL}$, and so resolving the identity with the orthonormal basis $(\hat{J}e_k)_{k\in K}$ yields the Kraus decomposition
\begin{equation*}\Theta(\om^o)(T):=(\om^o\ten\id)W^*(1\ten T)W=\sum_{k\in K}\hat{v}^*_kT\hat{v}_k,\hs\hs T\in\bd(\LT),\end{equation*}
where $\hat{v}_k=(\om_{\hat{J}\xi,\hat{J}e_k}\ten\id)(W)\in\hat{M}$. But $(\om_{\hat{J}\xi,\hat{J}e_k}\ten\id)(W)=(\om_{e_k,\xi}\ten\id)(W^*)$ by \cite[Proposition 2.4.6]{ES} (as unimodular quantum groups are Kac algebras), so that $\hat{v}_k=\hat{w}_k^*$. Hence, $\sum_{k\in K}\hat{w}_k\hat{w}_k^*=\sum_{k\in K}\hat{v}_k^*\hat{v}_k=\om^o(1)1=\om(1)1$.\end{proof}

For a unimodular locally compact quantum group $\G$, and $1\leq p \leq 2$, $\frac1p +\frac1q =1$, the noncommutative Fourier transform
$\mc{F}_p : \Lp \rightarrow  L^q(\hat\G)$
is the (unique) extension of the map $\mc{M}_\vphi \ni x \mapsto \lambda(\fee_x)\in \LLL$. The Hausdorff--Young inequality \cite[Theorem 3.2]{Coo} then states that
$\mc{F}_p$ is a contraction. Moreover, $\mc{F}:=\mc{F}_2$ is an isometric isomorphism of $\LT$ onto $ L^2(\hat{\G})$.

Given a positive $\rho\in\mc{T}(\LT)$ with $\tr(\rho)=1$, we let $D\in\LO_1^+$ be the density of $\rho|_{\LL}$, in the sense that $\tr(\rho x)=\vphi(Dx)$ for all $x\in\LL$. We also let $\hat{\rho}=\mc{F}\rho\mc{F}^*\in\mc{T}( L^2(\hat{\G}))$, and consider the associated density $\hat{D}\in L^1(\hat{\G})_1^+$.

\begin{thm}\label{theorem} Let $\G$ be a unimodular locally compact quantum group, and $\rho\in\mc{T}(\LT)$ be positive with $\tr(\rho)=1$. Then for $D\in\LO_1^+$ and $\hat{D}\in L^1(\hat{\G})_1^+$ as above satisfying $|H(D,\vphi)|,| H(\hat{D},\hat{\vphi})|<\infty$, we have
\begin{equation}\label{inq}H(D,\vphi)+ H(\hat{D},\hat{\vphi})\geq H(\rho,\tr).\end{equation}\end{thm}

\begin{proof} We follow along similar lines as in \cite{FL}. First consider the case when $D\in\Mphi$ and $\hat{D}\in\mc{M}_{\hat{\vphi}}$. Then
\begin{equation*}H(D,\vphi)+ H(\hat{D},\hat{\vphi})=-\tr(\rho\log D)-\tr(\rho\mc{F}^*\log\hat{D}\mc{F})=\tr(\rho A),\end{equation*}
where $A=-\log D-\mc{F}^*\log\hat{D}\mc{F}$. Letting $(e_i)_{i\in I}$ be an orthonormal basis of $\LT$ consisting of self-adjoint operators in $\Nphi$ for all $i\in I$, and $(\hat{e}_j)_{j\in J}$ be an orthonormal basis of $ L^2(\hat{\G})$ in $\N_{\hat{\vphi}}$, Lemma \ref{GT} then implies
\begin{align*}\tr(e^{-A})&\leq\tr(D^{1/2}\mc{F}^*\hat{D}\mc{F}D^{1/2})=\sum_{i\in I}\sum_{j\in J}|\la\hat{D}^{1/2}\mc{F}D^{1/2}e_i,\hat{e}_j\ra_{ L^2(\hat{\G})}|^2\\
                         &=\sum_{i\in I}\sum_{j\in J}|\la\mc{F}D^{1/2}e_i,\hat{D}^{1/2}\hat{e}_j\ra_{ L^2(\hat{\G})}|^2=\sum_{i\in I}\sum_{j\in J}|\la\lm(D^{1/2}e_i),\hat{D}^{1/2}\hat{e}_j\ra_{ L^2(\hat{\G})}|^2\\
                         &=\sum_{i\in I}\sum_{j\in J}|(\vphi_{D^{1/2}e_i}\ten\hat{\vphi}_{(\hat{D}^{1/2}\hat{e}_j)^*})(W)|^2=\sum_{i\in I}\sum_{j\in J}|(\vphi\ten\hat{\vphi})(W(D^{1/2}e_i\ten\hat{e}_j^*\hat{D}^{1/2}))|^2\\
                         &=\sum_{i\in I}\sum_{j\in J}|(\vphi\ten\hat{\vphi})((1\ten\hat{D}^{1/2})W(D^{1/2}\ten 1)(e_i\ten\hat{e}_j^*))|^2\\
                         &=\sum_{i\in I}\sum_{j\in J}|\la(1\ten\hat{D}^{1/2})W(D^{1/2}\ten 1),e_i\ten\hat{e}_j\ra_{\LT\ten L^2(\hat{\G})}|^2\\
                         &=\norm{(1\ten\hat{D}^{1/2})W(D^{1/2}\ten 1)}^2_{ L^2(\G)\ten L^2(\hat{\G})}=\hat{\vphi}((\vphi_D\ten\id)W^*(1\ten\hat{D})W).\end{align*}
Identifying $\hat{M}$ with $\LLL$, by Lemma \ref{Kraus} there exists a net $(\hat{w}_k)_{k\in K}$ in $ L^{\infty}(\hat{\G})$ satisfying \begin{equation*}(\vphi_D\ten\id)W^*(1\ten\hat{D})W=\sum_{k\in K}\hat{w}_k^*\hat{D}\hat{w}_k\hs\hs\text{and}\hs\hs\sum_{k\in K}\hat{w}_k^*\hat{w}_k=\sum_{k\in K}\hat{w}_k\hat{w}_k^*=\vphi(D)1,\end{equation*}
where all sums converge in the weak* topology of $\bd( L^2(\hat{\G}))$. Indexing by finite subsets $F$ of $K$, $\hat{D}_F:=\sum_{k\in F}\hat{w}_k^*\hat{D}\hat{w}_k$ defines a bounded increasing net of positive operators, so that $\hat{D}_F$ converges strongly to its supremum. Since the weak operator topology is equivalent to the weak* topology on bounded subsets of $\bd( L^2(\hat{\G}))$, it follows that $\sup_F\hat{D}_F=(\vphi_D\ten\id)W^*(1\ten\hat{D})W$. Thus, by normality of $\hat{\vphi}$
\begin{equation*}\hat{\vphi}((\vphi_D\ten\id)W^*(1\ten\hat{D})W)=\sup_F\hat{\vphi}\bigg(\sum_{k\in F}\hat{w}_k^*\hat{D}\hat{w}_k\bigg)=\sup_F\hat{\vphi}\bigg(\sum_{k\in F}\hat{D}\hat{w}_k\hat{w}_k^*\bigg)=\vphi(D)\hat{\vphi}(\hat{D})=1.\end{equation*}
Hence, $\tr(e^{-A})\leq1$ and Lemma \ref{Gibbs} yields
\begin{equation*}H(D,\vphi)+ H(\hat{D},\hat{\vphi})\geq H(\rho,\tr).\end{equation*}
In the general case, for $n\in\mathbb{N}$, we let $D_n:=\chi_{[0,n]}(D)\in\M_{\vphi}$ and $\hat{D}_n:=\chi_{[0,n]}(\hat{D})\in\M_{\hat{\vphi}}$. Then with $A_n:=-\log D_n-\mc{F}^*\log\hat{D}_n\mc{F}$, the above argument yields $\tr(e^{-A_n})\leq\vphi(D_n)\hat{\vphi}(\hat{D}_n)<\infty$. Thus, by monotonicity (see \cite[\S IX.2]{Tak2} for details)
\begin{align*}H(D,\vphi)+ H(\hat{D},\hat{\vphi})&=\lim_{n\rightarrow\infty}\bigg(-\vphi(D\log D_n)-\hat{\vphi}(\hat{D}\log\hat{D}_n)\bigg)\\
&\geq\lim_{n\rightarrow\infty}\bigg(H(\rho,\tr)-\log(\vphi(D_n)\hat{\vphi}(\hat{D}_n))\bigg)=H(\rho,\tr).\end{align*}
\end{proof}

\begin{remark} Using the theory of generalized $s$-numbers of measurable operators affiliated to semi-finite von Neumann algebras (cf. \cite{FK}), the above relation (\ref{inq}) becomes a classical inequality relating probability measures on $(0,\infty)$. Indeed, for any positive $\rho\in\mc{T}(\LT)$ with $\tr(\rho)=1$, the associated densities $D\in\LO$ and $\hat{D}\in L^1(\hat{\G})$ are positive self-adjoint operators affiliated to the semi-finite von Neumann algebras $\LL$ and $\LLL$, respectively. Denoting their spectral decompositions by $(e_\lm)$ and $(\hat{e}_\lm)$, their $t^{th}$ singular numbers are $\mu_t(D)=\inf\{s\geq0\mid\vphi(e_{(s,\infty)})\leq t\}$ and $\hat{\mu}_t(\hat{D})=\inf\{s\geq0\mid\hat{\vphi}(\hat{e}_{(s,\infty)})\leq t\}$, for $t>0$. These form probability densities on $(0,\infty)$ satisfying
\begin{equation*}H(D,\vphi)=-\int_{0}^{\infty}\mu_t(D)\log\mu_t(D)dt\hs\hs\text{and}\hs\hs H(\hat{D},\hat{\vphi})=-\int_0^{\infty}\hat{\mu}_t(\hat{D})\log\hat{\mu}_t(\hat{D})dt\end{equation*}
by \cite[Remark 3.3]{FK}, where $dt$ denotes the Lebesgue measure. Thus, in this setting it appears that the generalized singular numbers of operators and their non-commutative Fourier transforms behave in a similar manner to the classical Fourier transforms of functions.\end{remark}

\section{Special Cases and Complementarity}

A locally compact quantum group $\G$ is said to be \e{compact} if $\vphi$ is finite. We say that a compact quantum group is of \e{Kac type} if $\vphi$
is a tracial state. The representation theory for such quantum groups closely parallels that of compact groups (cf. \cite{ES,W}),
and we shall use this theory to elucidate Theorem \ref{theorem} in this setting. Our result may be seen as a generalization of \cite[Theorem 2]{AR}
- the strongest known quantitative uncertainty principle for arbitrary compact groups - to the setting of compact quantum groups of Kac type. We also show that finite-dimensional quantum groups give rise to canonical complementary subalgebras in the sense of Petz (cf. \cite{Petz}). We begin with a short review of the necessary tools from representation theory. Our reference throughout is \cite{ES}.

Let $\G = (M, \Gamma, \fee, \psi)$ be a compact quantum group of Kac type. In this case $M_*$ becomes an involutive Banach algebra. By a \e{representation} of $M_*$, we therefore mean a $\ast$-homomorphism $\alpha: M_*\rightarrow\bd(H_\alpha)$.
We assume the reader is familiar with the notions of irreducibility, non-degeneracy, and unitary equivalence for representations of involutive Banach algebras.
We have the following facts about $\G$: every irreducible representation $\alpha$ of $M_*$ is finite-dimensional and is unitarily equivalent to a sub-representation of the left regular representation $\lm$ with multiplicity $d_\alpha:=\dim H_\alpha$, and every non-degenerate representation of $M_*$ can be decomposed into a direct sum of irreducible representations. Thus, $\lm$ is unitarily equivalent to the direct sum (of equivalence classes) of irreducible representations $\alpha$, each occurring with multiplicity $d_\alpha$. We remark that every irreducible $\alpha$ has a unitary generator $u^{\alpha}\in\LL\ten M_{d_\alpha}(\C)$ satisfying $\alpha(f)=(f\ten\id)(u^{\alpha})$ for $f\in M_*$ and
\begin{equation*}\Gam(u^{\alpha}_{ij})=\sum_{k=1}^{d_\alpha}u^{\alpha}_{ik}\ten u^{\alpha}_{kj}\end{equation*}
for all $1\leq i,j\leq d_{\alpha}$. The generator $u^{\alpha}$ is called a \e{unitary co-representation} of $\G$.

As in the group case, $\LT$ becomes a Banach algebra, and there exists a continuous homomorphism $b:\LT\rightarrow M_*$, which, under the canonical identification (\ref{je1}), is the inclusion of $\LT$ into $\LO$.

In light of the above, if $I$ denotes the set of equivalence classes of irreducible representations, it follows that the Fourier transform
\begin{equation}\label{je2}
\mc{F}:\LT\ni x \,\rightarrow\,\bigoplus_{\alpha\in I}\,\alpha(b(x))\in\,{\tiny{\ell^2-}}\bigoplus_{\alpha\in I}\,\mc{HS}(H_\alpha)
\end{equation}
is a Hilbert space isomorphism, where the norm in ${\tiny{\ell^2}-}\bigoplus_{\alpha\in I}\mc{HS}(H_\alpha)$ is given by
\begin{equation*}
\bignorm{\,\bigoplus_{\alpha\in I}\alpha(b(x))\,}^2 \,=\, \sum_{\alpha\in I}\, d_\alpha\cdot\tr_{\alpha}\big(\,\alpha(b(x))^* \, \alpha(b(x))\,\big)\,,
\end{equation*}
Above, $\tr_\alpha$ denotes the unnormalized trace on $M_{d_\alpha}(\C)$, and the factor of $d_\alpha$ accounts for the multiplicity of $\alpha$ in the left regular representation. Under the identification (\ref{je2}) we have $\ell^\infty(\hat\G)=\tiny{\ell^\infty-}\bigoplus_{\alpha\in I}\,\mc{B}(H_\alpha)\ten1_{d_\alpha}$. In the case of compact groups, for $f\in L^2(G)$, the matrix $\alpha(b(f))=\hat{f}(\alpha)^{\mathfrak t}$, where $\hat{f}(\alpha)=\int_Gf(s)\alpha(s)^*ds$ and $\mathfrak t$ denotes the transpose.

If $\rho\in\mc{T}(\LT)$ is positive, then as above we define $D\in\LO$ to be the density associated to $\rho|_{\LL}$, and $\hat{D}\in\ell^1(\hat\G)$ to be the density associated to $\mc{F}\rho\mc{F}^*|_{\ell^{\infty}(\hat{\G})}$, which in this case is given by the direct sum $\hat{D}=\oplus_{\alpha}\hat{D}^{\alpha}$ satisfying $\sum_{\alpha}d_{\alpha}\tr_{\alpha}(\hat{D}^{\alpha})=\tr(\rho)$.

\begin{thm}\label{compact} Let $\G$ be a compact quantum group of Kac type, and $\rho\in\mc{T}(\LT)$ be positive and non-zero. Then for $D\in\LO$ and $\hat{D}\in\ell^1(\hat{\G})$ as above, we have
\begin{equation}\label{upcompact}\vphi(s(D))\,\bigg(\sum_{\alpha\in I}\,d_\alpha\cdot\textnormal{rank}(\hat{D}^{\alpha})\bigg) \,\geq \,e^{H\left(\frac{\rho}{\tr(p)},\tr\right)}\,, \end{equation}
where $s(D)$ is the support projection of $\vphi_D\in M_*^+$. In particular, for $x\in\LT$, $x\neq0$,
\begin{equation*}\vphi(s(|x|^2))\,\bigg(\sum_{\alpha\in I}\,d_\alpha\cdot\textnormal{rank}(\alpha(b(x))\bigg) \,\geq 1\end{equation*}\end{thm}

\begin{proof} By dividing through by $\tr(\rho)$, without loss of generality we may assume that $\rho$ is state. We first claim that $H(D,\vphi) \leq \log(\vphi(s(D)))$. To show this, we use the inequality $\alpha - \alpha \log(\alpha)\leq \beta - \alpha\log(\beta)$
for $\alpha,\beta\in(0,\infty)$. By functional calculus, this implies
\begin{equation*} D - D \log(D)\leq\frac{1}{\vphi(s(D))} - D \log\left(\frac{1}{\vphi(s(D))}\right)\,.\end{equation*}
Since $\fee(s(D) D) = 1$, applying the positive normal linear functional $\vphi(s(D)\cdot)$ to both sides of the above inequality yields
$H(D,\vphi)\leq\log(\vphi(s(D)))$, which is our claim.

Next, using representation (\ref{je2}) we obtain
\begin{equation*}
 H(\hat{D},\hat{\vphi})=-\sum_{\alpha\in I}\,d_\alpha\cdot\tr_{\alpha}\bigg(\,\hat{D}^{\alpha} \, \log\big(\hat{D}^{\alpha}\big)\,\bigg)=H\bigg(\,\bigoplus_{\alpha\in I}\,\hat{D}^{\alpha}\,\oplus\,\stackrel{d_\alpha}{\cdots}\,\oplus\,\hat{D}^{\alpha}\,,\tr\bigg)\,.\end{equation*}
Since $H(A,\tr)\leq\log(\textnormal{rank}(A))$ for any density matrix $A$, it follows that
\begin{equation*}
 H(\hat{D},\hat{\vphi}) \,\leq\, \log\bigg(\,\sum_{\alpha\in I} \, d_\alpha\cdot\textnormal{rank}\big(\hat{D}^{\alpha})\big)\,\bigg)\,.
\end{equation*}
Putting things together, applying Theorem \ref{theorem} and exponentiating, we obtain
\begin{equation*}
\vphi(s(D))\,\bigg(\,\sum_{\alpha\in I} \, d_\alpha\cdot\textnormal{rank}(\hat{D}^{\alpha})\,\bigg)\,\geq\, e^{H(\rho,\tr)}\,.\end{equation*}
The final statement follows from the observation that for $\rho=\om_x$ with $x\in\LT$, we have $\hat{D}=\oplus_{\alpha\in I} \alpha(b(x))\,\alpha(b(x))^*$. \end{proof}

As an immediate corollary, we obtain a strengthening of \cite[Theorem 2]{AR}:

\begin{cor}Let $G$ be a compact group with normalized Haar measure $\mu_G$, and let $\rho\in\mc{T}( L^2(G))$ be positive and non-zero. Then for $D\in L^1(G)$ and $\hat{D}\in\oplus_{\alpha\in I}\mc{T}(H_\alpha)$ as above, we have
\begin{equation*}\mu_G \,(\textnormal{supp}(D))\,\bigg(\,\sum_{\alpha\in I}\,d_\alpha\cdot\textnormal{rank}(\hat{D}^{\alpha})\,\bigg)\,\geq\,e^{H\left(\frac{\rho}{\tr(p)},\tr\right)}\,.\end{equation*}
In particular, if $f\in L^2(G)$ is non-zero, then
\begin{equation*}
\mu_G \,(\textnormal{supp}(f))\,\bigg(\,\sum_{\alpha\in I}\,d_\alpha\cdot\textnormal{rank}(\hat{f}(\alpha))\,\bigg)\,\geq\,1\,.
\end{equation*}\end{cor}

The simplest case of a compact quantum group of Kac type is when $\G$ is finite, i.e., $\LL$ is finite-dimensional. In this setting the Haar weight
$\vphi$ on $\LL$ is the restriction of the canonical trace on $\bd(\LT)$ to $\LL$, and if we view $\LLL$ as a subalgebra of $\bd(\LT)$ (via conjugation with the Fourier transform), the dual weight $\hat{\vphi}$ on $\LLL$ is the restriction of the normalized trace $\hat{\vphi}=(\dim\G)^{-1}\tr$ (cf. \cite{ES}). For a state $\rho\in\mc{T}(\LT)$, one may easily verify that the respective densities of $\rho$ and $\hat{\rho}$ are given by $D=E(\rho)$ and $\hat{D}=(\dim\G)\cdot\mc{F}\hat{E}(\rho)\mc{F}^*$, where $E:\bd(\LT)\rightarrow\LL$ and $\hat{E}:\bd(\LT)\rightarrow\LLH$ are the unique trace-preserving conditional expectations onto $\LL$ and $\LLH$, respectively. Theorem \ref{theorem} then takes the following form.

\begin{cor}\label{corollary}
Let $\G$ be a finite-dimensional quantum group. Then for any state $\rho\in\mc{T}(\LT)$ we have
\begin{equation}\label{finite}H(E(\rho),\tr) \,+\, H(\hat{E}(\rho),\tr) \,\geq\, H(\rho,\tr) + \log(\dim\G)\,.\end{equation}\end{cor}

\begin{proof} On the one hand, since $\fee = \tr$ on $\LL$, we see that $H(D,\vphi)$ coincides with $H(E(\rho),\tr)$. On the other hand, we obtain
\begin{equation*}H(\hat{E}(\rho),\tr)=H\bigg(\frac{\mc{F}^*\hat{D}\mc{F}}{\dim\G},\tr\bigg)=H\bigg(\frac{\hat{D}}{\dim\G},\tr\bigg)=\log(\dim\G)+ H(\hat{D},\hat{\vphi}).\end{equation*}
Putting things together and applying Theorem \ref{theorem} yields the result.\end{proof}

\begin{remark} If $G$ is a finite group and $\rho\in\mc{T}( L^2(G))$ is a state, then inequality (\ref{finite}) reads
\begin{equation*}H(\mu_\rho)+H\bigg(\frac{C_\rho}{|G|},\tr\bigg)\geq H(\rho,\tr)+\log|G|,\end{equation*}
where $H(\mu_\rho)=-\sum_{s\in G}\la\rho\delta_s,\delta_s\ra\log(\la\rho\delta_s,\delta_s\ra)$, and $C_\rho$ is the correlation matrix associated
to the positive definite function $\psi_\rho(s)=\tr(\rho\lm(s))$, i.e., the $s,t$ entry of $C_\rho$ is $\psi_\rho(s^{-1}t)$.
\end{remark}

If $\G$ is a finite-dimensional quantum group such that $\LL$ and $\LLH$ are both commutative, then $\LL=\ell^{\infty}(G)$ and $\LLH= L(G)$ for a finite abelian group $G$, in which case $E$ and $\hat{E}$ are given by
\begin{equation*}
E(\rho) = \sum_{s\in G}\la\rho\delta_s,\delta_s\ra \,|\delta_s\ra\la\delta_s|
\hs\hs\hs\text{and}\hs\hs\hs
\hat{E}(\rho) = \frac{1}{|G|^2}\,\sum_{s\in G}\la\rho\chi^s,\chi^s\ra\,|\chi^s\ra\la\chi^s|
\end{equation*}
for any $\rho\in\mc{T}( L^2(G))$, where $\chi^s$ is the character on $G$ represented by $s\in G$. Inequality (\ref{finite}) then simply expresses the well-known fact that the orthonormal bases $\{\delta_s\mid s\in G\}$ and $\{|G|^{-1/2}\chi^s\mid s\in G\}$ are \e{mutually unbiased}. Since these are the canonical examples of such bases, and their complementary nature relies on abelian group duality, one may view inequality (\ref{finite}) as an extension of this complementarity to finite-dimensional quantum group duality. In fact, $\LL$ and $\LLH$ are \e{complementary subalgebras} in the sense of Petz (cf. \cite{Petz})
for any finite $\G$, where two subalgebras $A,B$ of $M_n(\C)$ are complementary if one of the following equivalent conditions are satisfied, where $\tau$ denotes the normalized trace on $M_n(\C)$:
\begin{itemize}\label{items}
\item If $p\in A$ and $q\in B$ are minimal projections, then $\tau(pq)=\tau(p)\tau(q)$;
\item $A\ominus\C1$ and $B\ominus\C1$ are orthogonal in $M_n(\C)$;
\item $\tau(ab)=\tau(a)\tau(b)$ for all $a\in A$, $b\in B$;
\item $E_B(A)=\C1$, where $E_B$ is the trace-preserving conditional expectation onto $B$.\end{itemize}

The fact that $\LL$ and $\LLH$ are complementary follows from a standard argument, which we now provide for the convenience of the reader. We also note that
the above concept of ``orthogonality'' of subalgebras was studied by Popa \cite{Popa83} in the setting of finite von Neumann algebras,
but we shall stick with the terminology of complementarity.

\begin{prop} Let $\G$ be a finite-dimensional quantum group. Then $\LL$ and $\LLH$ are complementary subalgebras of $\bd(\LT)$ such that $\la\LL\LLH\ra=\bd(\LT)$, where $\la\cdot\ra$ denotes linear span.\end{prop}

\begin{proof} The unique trace-preserving conditional expectation $\hat{E}:\bd(\LT)\rightarrow\LLH$ is given by
\begin{equation*}\hat{E}(x)=\frac{1}{\dim\G}(\id\ten\vphi)V(x\ten1)V^*,\hs\hs x\in\bd(\LT),\end{equation*}
where $V\in L^\infty(\hat\G)'\ten\LL$ is the right fundamental unitary of $\G$ (cf. \cite{KV}). Indeed, if $\Gam^r:\bd(\LT)\rightarrow\bd(\LT)\ten\LL$ denotes the map $\Gam^r(x)=V(x\ten1)V^*$, $x\in\bd(\LT)$,
then $\Gam^r$ is the co-associative co-multiplication on $\bd(\LT)$ obtained by (right) extension of $\Gam$, and $\hat{E}$ is the extension of the left
convolution action of the Haar weight on $\LL$ to $\bd(\LT)$ (cf. \cite{JNR}). Clearly, $\hat{E}$ is a complete contraction, and %since $V\in\LLH\ten\LL$ we have
$\hat{E}(\hat{x})=\hat{x}$ for $\hat{x}\in\LLH$. On the other hand, by left invariance of $\vphi$, which means $(\id\ten\vphi)\Gam(\cdot)=\vphi(\cdot)1$ on $\LL$,
we see that for any $x\in\bd(\LT)$,
\begin{eqnarray*}
\Gam^r(\hat E(x))
&=&\frac{1}{\dim\G}\Gam^r((\id\ten\vphi)\Gam^r(x))\\
&=& \frac{1}{\dim\G}(\id\ten\id\ten\vphi)(\Gam^r\ten\id)(\Gam^r(x))\\
&=& \frac{1}{\dim\G}(\id\ten\id\ten\vphi)(\id\ten\Gam^r)(\Gam^r(x))\hs\hs\hs\hs(\textnormal{co-associativity})\\
&=& \frac{1}{\dim\G}(\id\ten\id\ten\vphi)(\id\ten\Gam)(\Gam^r(x))\hs\hs\hs\hs(\Gam^r(x)\in\bd(\LT)\ten\LL)\\
&=& \frac{1}{\dim\G}(\id\ten\vphi)(\Gam^r(x))\ten1\hs\hs\hs\hs(\textnormal{left invariance})\\
&=& \hat E(x)\ten 1\,.
\end{eqnarray*}
Thus, $V(\hat E(x)\ten1)=(\hat E(x)\ten1)V$, and applying the slice map $(\id\ten\om)$ to both sides of this equation yields $\rho(\om) \hat E(x)= \hat E(x)\rho(\om)$, for all $\om\in\LO$.
Then, by weak* density of $\rho(\LO)$ in $ L^\infty(\hat\G)'$, we have $E(x)\in L^\infty(\hat\G)''=\LLH$. Thus $\hat{E}$ is a projection of norm one onto $\LLH$. That it is also trace-preserving is clear.

Now, for $x\in\LL$, we have $\displaystyle \hat{E}(x)=\frac{1}{\dim\G}(\id\ten\vphi)\Gam(x)=\frac{1}{\dim\G}\vphi(x)\in\C$ by left invariance, so
the subalgebras $\LL$ and $\LLH$ are complementary. The fact that $\la\LL\LLH\ra=\bd(\LT)$ follows from the general relation
at the level of locally compact quantum groups (cf. \cite[Proposition 2.5]{VV}).\end{proof}

Another question of interest is an entropic characterization of complementarity. For certain classes of subalgebras $A,B$ of $M_n(\C)$, it was shown that the
maximality of the conditional entropy of Connes--St\o rmer (cf. \cite{CS}) is equivalent to complementarity \cite{PSW}. Also, it was recently shown by Choda
that for finite-dimensional subalgebras $A\cong M_n(\C)$ inside a finite von Neumann algebra $M$, the complementarity of $A$ and $uAu^*$ for some unitary $u\in M$ can be characterized by the maximal entropy of a certain density matrix related to $u$ \cite{Choda}. It would be interesting to see whether inequality (\ref{finite}) is necessary/sufficient for complementarity. This is certainly true when $A$ and $B$ are maximal abelian (cf. \cite{P}).

\section{Non-unimodular Setting}

Our aim in this section is to put forth the idea that the relative entropy with respect to the Haar weight is the appropriate candidate for studying entropic properties of general quantum measures. This is primarily justified by providing an uncertainty principle for arbitrary locally compact groups. Interestingly, as we shall see, the non-unimodularity adds an intrinsic degree of freedom to the overall uncertainty. Further justification is provided by our last result, which states that the relative entropy restricted to $q$-traces of discrete quantum groups reduces to the entropy of Hiai--Izumi \cite{HI}, which is crucial for studying the dynamics of the corresponding random walk. We begin with the necessary preliminaries from the spatial theory of von Neumann algebras. For details we refer the reader to \cite{C}.

Let $M$ be a von Neumann algebra on a Hilbert space $H$, and let $\psi$ be a fixed normal semi-finite faithful weight $M'$. A vector $\xi\in H$ is $\psi$-bounded
if the mapping $R^{\psi}(\xi):\N_\psi\ni\Lm_{\psi}(x')\mapsto x'\xi\in H$ extends to a bounded linear operator from $H_\psi$ into $H$. We denote by
$\mc{D}(H,\psi)$ the set of $\psi$-bounded vectors. It follows that $R^{\psi}(\xi)R^{\psi}(\xi)^*\in M$ for all $\xi\in\mc{D}(H,\psi)$.
Then for any normal semi-finite weight $\vphi$ on $M$, the \e{spatial derivative} $d\vphi/d\psi$ is the largest positive self-adjoint operator $T$
on $H$ satisfying
\begin{equation*}
\vphi(R^{\psi}(\xi)R^{\psi}(\xi)^*)=\begin{cases}\bignorm{T^{1/2}\xi}^2_2 &\text{if}\hs\xi\in\mc{D}(H,\psi)\cap\mc{D}(T^{1/2}),\\
                                        +\infty &\text{otherwise}.\end{cases}
                                        \end{equation*}
If $\vphi$ were bounded, so that $\vphi(1)<\infty$, then $\mc{D}(H,\psi)$ is contained in the domain of $(d\vphi/d\psi)^{1/2}$ and is a core for this operator.

Now, let $\G=(\LL,\Gam,\vphi,\psi)$ be an arbitrary locally compact quantum group. We denote $M_*$ by $\LO$ and $H_\vphi$ by $\LT$. Since $\LL$ is standardly presented on $\LT$, every state $\om\in\LO$ is the restriction of a vector state $\om_\xi$ to $\LL$, for some $\xi\in\LT$. Thus, for a state $\om\in\LO$, we define its \e{entropy} to be
\begin{equation}H(\om,\vphi):=\begin{cases}-\bigg\la\log\bigg(\frac{d\om_\xi'}{d\vphi}\bigg)\xi,\xi\bigg\ra &\text{if}\hs\xi\in\mc{D}\bigg(\log\bigg(\frac{d\om_\xi'}{d\vphi}\bigg)\bigg),\\
                                        +\infty &\text{otherwise},\end{cases}\end{equation}
where $\om_\xi'=\om_{\xi}|_{\LL'}$. By properties of the spatial derivative, this definition is independent of the representing vector $\xi$, and is equal to $-S(\om,\vphi)$, where $S(\om,\vphi)$ is the \e{relative entropy} of $\om$ and the left Haar weight $\vphi$ (cf. \cite[\S5]{OP}). For later purposes we note that $H(\om_\xi,\vphi)=H(\om_{J\xi}',\vphi')$, where
\begin{equation*}H(\om_{J\xi},\vphi')=-\bigg\la\log\bigg(\frac{d\om_{J\xi}}{d\vphi'}\bigg)J\xi,J\xi\bigg\ra,\end{equation*}
$\vphi'$ is the normal semi-finite faithful weight on $\LL'$ given by $\vphi'(x')=\vphi(Jx'^*J)$, for $x'\in\LL'^+$, and $J$ is the anti-linear isometry associated to the standard representation of $\LL$.

To get a sense of what these spatial derivatives look like, let $G$ be a locally compact group with left Haar measure $\mu_G$, viewed as a weight on $L^{\infty}(G)$ via integration. Then for $f\in L^2(G)$ with $\norm{f}_2=1$, $d\om_f/d\mu_G$ is the (possibly unbounded) operator of multiplication by $|f|^2$ on $L^2(G)$, and
\begin{equation*}
H(\om_f,\mu_G) \,=\, -\bigg\la\log\bigg(\frac{d\om_f}{d\mu_G}\bigg)\,f \,,\, f\bigg\ra \,=\, -\int_G\abs{f(s)}^2\log\abs{f(s)}^2d\mu_G(s)\,
\end{equation*}
when $f\in\mc{D}(M_{|f|^2})$. In particular, if $G$ is compact, $-H(\om_f,\mu_G)$ is the Kullback--Leibler divergence of the probability density $|f|^2$, and $\mu_G$.

On the dual side, let $\hat{\vphi}$ and $\hat{\vphi}'$ denote the Plancherel weight on $VN(G)$ and $\RG=VN(G)'$, respectively (cf. \cite{Tak2}). The spatial derivative $d\om_{Jf}/d\hat{\vphi}'$ is then related to the Fourier transform of $f$, given by $\F(f)\xi=f\ast\Delta^{1/2}\xi$, $\xi\in\mc{D}(\F(f))$ (cf. \cite{T}). Indeed, for $f\in C_c(G)$ one can show that $|\F(f)|^2=d\om_{Jf}/d\hat{\vphi}'$, in which case if $\norm{f}_2=1$, we have
\begin{equation*}
 H(\om_f,\hat{\vphi})= H(\om_{Jf},\hat{\vphi}')=-\left\la\log\big(|\F(f)|^2\big)Jf,Jf\right\ra.\end{equation*}

Of course, if $\G$ were a unimodular locally compact quantum group and $\om\in\LO$ were a state with density $D$, then $H(\om,\vphi)$ defined above coincides with $H(D,\vphi)$ from \S1 (see \cite[\S2]{T} for details).

We shall now present a partial generalization of Theorem \ref{theorem} to the case of vector states on locally compact groups. For this, recall that the right regular representation of a locally compact group $G$ is defined by $\rho(g)\xi(s)=\xi(sg)\Delta(g)^{1/2}$, for $g,s\in G$ and $\xi\in L^2(G)$, and also integrates to a non-degenerate $\ast$-representation of $L^1(G)$ in the usual manner.

\begin{thm} Let $G$ be a locally compact group with left Haar measure $\mu_G$, let $\vphi$ be the Plancherel weight on $VN(G)$, and let $\xi\in L^2(G)$ with $\norm{\xi}_2=1$. If $H(\om_\xi,\mu_G)$ and $ H(\om_{\mc{F}\xi},\hat{\vphi})$ are finite, then
\begin{equation}\label{Delta} H(\om_\xi,\mu_G) +  H(\om_{\mc{F}\xi},\hat{\vphi}) \geq -\log\norm{\Delta^{-1/2}\xi}_2^2,\end{equation}
where for $\xi\notin\mc{D}(\Delta^{-1/2})$ we let $\norm{\Delta^{-1/2}\xi}_2=\infty$.\end{thm}

\begin{proof} Throughout the proof we view $VN(G)$ as a subalgebra of $\bd( L^2(G))$, and we distinguish between the various Hilbert space representations of $VN(G)$. In particular, the Fourier transform defined above is a unitary isomorphism $\mc{F}: L^2(G)\rightarrow L^2(VN(G),\vphi')$, where $ L^2(VN(G),\vphi')$ is the spatial non-commutative $ L^2$-space in the sense of Hilsum \cite{H}. The latter space is also unitarily equivalent to $H_\vphi$ via $\beta(\lm(f)\Delta^{1/2})=\Lm_{\vphi}(\lm(f))$ for $\lm(f)\in\N_\vphi$ \cite[Theorem 23]{T2}, which in turn is unitarily equivalent to $ L^2(G)$ via $\alpha(\Lm_{\vphi}(\lm(f)))=f$ for $\lm(f)\in\N_\vphi$. In all, $\beta\circ\mc{F}\circ\alpha=\id_{H_\vphi}$ (cf. \cite{Coo}). We let $\tilde{\vphi}:=\beta^*\alpha^*\cdot\vphi\cdot\alpha\beta$ be the conjugate weight of $\vphi$.

Now, for $n\in\mathbb{N}$, we let $\xi_n(s)=|\xi(s)|$ when $|\xi(s)|^2\leq n$ and $\xi_n(s)=0$ otherwise, so that $M^2_{\xi_n}=\chi_{[0,n]}(M_{|\xi|^2})$, and we let $x_n=\chi_{[0,n]}(d\om_{\mc{F}\xi}/d\tilde{\vphi})$. Then
\begin{equation*}-\la\log(M^2_{\xi_n})\xi,\xi\ra - \la\log(x_n)\mc{F}\xi,\mc{F}\xi\ra=\tr(|\xi\ra\la\xi| A_n),\end{equation*}
where $A_n:=-\log(M^2_{\xi_n})-\mc{F}^*\log(x_n)\mc{F}$. Next, Lemma \ref{GT} yields
\begin{equation*}\tr(e^{-A_n})\leq\tr(M_{\xi_n}\mc{F}^*x_n\mc{F}M_{\xi_n})=\sum_{i\in I}\la x_n\mc{F}(\xi_ne_i),\mc{F}(\xi_ne_i)\ra,\end{equation*}
where $(e_i)_{i\in I}$ is an orthonormal basis of $ L^2(G)$ consisting of non-negative continuous functions with compact support. Since $\xi_ne_i$ is a bounded, compactly supported function in $ L^1(G)$, the vector $\alpha(\Lm_{\vphi}(\lm(\xi_ne_i)))\in L^2(G)$ is $\vphi$-bounded, and it follows that $R^{\vphi}(\alpha(\Lm_{\vphi}(\lm(\xi_ne_i))))=\rho(\Delta^{-1/2}\check{\xi_n}\check{e_i})\circ\alpha$, where $\check{\xi_n}\check{e_i}(s)=\xi_ne_i(s^{-1})$ for $s\in G$. Thus, by properties of the spatial derivative we get
\begin{align*}\bigg\la x_n\mc{F}(\xi_ne_i),\mc{F}(\xi_ne_i)\bigg\ra&\leq\bigg\la\beta\frac{d\om_{\mc{F}\xi}}{d\tilde{\vphi}}\beta^*\Lm_{\vphi}(\lm(\xi_ne_i)), \Lm_{\vphi}(\lm(\xi_ne_i))\bigg\ra\\                             &=\bigg\la\frac{d\om_{\xi}}{d\vphi}\alpha(\Lm_{\vphi}(\lm(\xi_ne_i))),\alpha(\Lm_{\vphi}(\lm(\xi_ne_i)))\bigg\ra\\
&=\om_\xi(R^{\vphi}(\alpha(\Lm_{\vphi}(\lm(\xi_ne_i))))R^{\vphi}(\alpha(\Lm_{\vphi}(\lm(\xi_ne_i))))^*)\\
&=\bigg\la\rho(\Delta^{-1/2}\check{\xi_n}\check{e_i})\rho(\Delta^{-1/2}\check{\xi_n}\check{e_i})^*\xi,\xi\bigg\ra\leq\norm{\Delta^{-1/2}\check{\xi_n}\check{e_i}}_1^2.\end{align*}
But
\begin{equation*}\norm{\Delta^{-1/2}\check{\xi_n}\check{e_i}}_1=\int_G\Delta(s)^{-1/2}\xi_n(s^{-1})e_i(s^{-1})ds=\int_G\Delta(s)^{-1/2}\xi_n(s)e_i(s)ds=\la\Delta^{-1/2}\xi_n,e_i\ra,\end{equation*}
and since this is true for arbitrary $i\in I$, we have
\begin{equation*}\tr(e^{-A_n})\leq\sum_{i\in I}|\la\Delta^{-1/2}\xi_n,e_i\ra|^2=\norm{\Delta^{-1/2}\xi_n}_2^2<\infty.\end{equation*}
Thus, applying Lemma \ref{Gibbs} to the pure state $\om_\xi$, we obtain
\begin{equation*}-\la\log(M^2_{\xi_n})\xi,\xi\ra - \la\log(x_n)\mc{F}\xi,\mc{F}\xi\ra\geq -\log\norm{\Delta^{-1/2}\xi_n}_2^2.\end{equation*}
Finally, since $\Delta^{-1/2}\xi_n$ increases pointwise to $\Delta^{-1/2}|\xi|$, by monotonicity we have
\begin{equation*}H(\om_\xi,\mu_G) +  H(\om_{\mc{F}\xi},\hat{\vphi})=\lim_{n\rightarrow\infty}\bigg(-\la\log(M^2_{\xi_n})\xi,\xi\ra - \la\log(x_n)\mc{F}\xi,\mc{F}\xi\ra\bigg)\geq-\log\norm{\Delta^{-1/2}\xi}_2^2.\end{equation*}
\end{proof}

In the theory of random walks on discrete groups, entropic quantities play a significant role in describing the probabilistic behavior (cf. \cite{Kai-Ver}). With the emergence of non-commutative random walks on discrete quantum groups \cite{I}, it is natural to ask whether entropic quantities can be used to study the
corresponding dynamics. This has been done, for example, in \cite{HI}, where amenability of fusion algebras was studied via entropies of random walks generated
by ``$q$-traces''. We now show that our entropy reduces to the entropy of \cite{HI} when restricted to $q$-traces.

Let $\G$ be a discrete quantum group, i.e., the dual $\hat{\G}$ is compact. In this case we may identify
\begin{equation*}
\Ll \,\cong\, \bigoplus_{\alpha\in I} \, M_{d_\alpha}(\C)\,,
\end{equation*}
where the direct sum is taken over all irreducible unitary co-representations of $\hat{\G}$. In the case of compact groups, $VN(G)$ becomes a discrete quantum
group and the above decomposition is that arising from the Peter--Weyl theorem.

For every $\alpha\in I$ there exists a positive invertible matrix
$F^\alpha\in M_{d_\alpha}(\C)$ such that the corresponding ``$F$--matrices'' implement the left Haar weight $\vphi$ in the sense that
\begin{equation*}
\vphi(x)\,=\,\sum_{\alpha\in I}\,\tr(F^\alpha)\,\tr(F^\alpha x),\hs\hs x\in\Mphi\,.
\end{equation*}
Given a state $\om\in\Lo$, let $\sum_{\alpha\in I}D_\alpha\in\Ll$ denote the density of $\om$ via trace duality. It then follows that
$D_\om:=\sum_{\alpha\in I}\tr(F^\alpha)^{-1}(F^\alpha)^{-1}D_\alpha$ is the density of $\om$ with respect to the Haar weight, i.e.,  $\om(x)=\vphi(D_\om x)$ for $x\in\Ll$. If we restrict our attention to so-called ``$q$-traces'' of the form
$\mu=\sum_{\alpha}\mu_\alpha\delta_\alpha$, where $\delta_\alpha(\cdot)=\tr(F^\alpha)^{-1}\tr(F^\alpha(\cdot))$ and $(\mu_\alpha)\in\ell^1(I)$ is a probability measure, then $D_\mu=\sum_{\alpha}\tr(F^\alpha)^{-2}\mu_\alpha z_\alpha$, where $z_\alpha\in\mc{Z}(\Ll)$ is the central projection corresponding to the factor $M_{d_\alpha}$. We now show that for such states $\mu$, our entropy $H(\mu,\vphi)$ coincides with $H_\sigma(\mu)$ as defined in \cite[\S2]{HI}.

\begin{thm}\label{last} Let $\G$ be a discrete quantum group, and $\mu=\sum_{\alpha}\mu_\alpha\delta_\alpha$ for some probability distribution $(\mu_\alpha)\in\ell^1(I)$. Then if $H(\mu,\vphi)$ is finite we have
\begin{equation*}H(\mu,\vphi)=-\sum_{\alpha\in I}\mu_\alpha\log\bigg(\frac{\mu_\alpha}{\tr(F^\alpha)^2}\bigg)=H_\sigma(\mu).\end{equation*}\end{thm}

\begin{proof} Clearly $\mu=\om_{\Lm_{\vphi}(D_\mu^{1/2})}|_{\Ll}$, and $J\Lm_{\vphi}(D_\mu^{1/2})=\Lm_{\vphi}(D_\mu^{1/2})$ as $D_\mu\in\mc{Z}(\Ll)$. Thus, $H(\mu,\vphi)=H(\om_{\Lm_{\vphi}(D_\mu^{1/2})},\vphi)=H(\om_{\Lm_{\vphi}(D_\mu^{1/2})}',\vphi')$, so we shall use the spatial derivative with respect to the Haar weight on the commutant given by $\vphi'(x')=\vphi(Jx'^*J)$, for $x'\in\Ll'^+$.

We first consider the case of finitely supported $\mu$, say $\mu=\sum_{\alpha\in F}\mu_\alpha\delta_\alpha$, with $|F|<\infty$. If $z_F=\sum_{\alpha\in F}z_\alpha$ is the central projection of $\ell^\infty(\G)$ corresponding to the support of $D_\mu$, it readily follows from the definition of the spatial derivative that
\begin{equation*}\frac{d\mu}{d\vphi'}=\frac{d\mu}{d\vphi_F'},\end{equation*}
where $\vphi_F(x)=\vphi(z_Fx)$ for $x\in\ell^{\infty}(\G)$ and $\vphi_F'(x')=\vphi_F(Jx'^*J)$ for $x'\in\ell^{\infty}(\G)'$. If $E:\ell^{\infty}(\G)\rightarrow\mc{Z}(\ell^{\infty}(\G))$ denotes the normal faithful conditional expectation $E(x)=\sum_{\alpha\in I}\tr(F^\alpha)^{-1}\tr(F^\alpha x)z_\alpha$, $x\in\ell^{\infty}(\G)$, we have $\vphi_F\circ E=\vphi_F$. Thus, by the unnormalized version of \cite[Theorem 5.15]{OP} (cf. \cite[Proposition 5.1]{OP}) we have
\begin{equation*}H(\mu,\vphi)=-S(\mu,\vphi)=-S(\mu,\vphi_F)=-S(\mu|_{\mc{Z}(\ell^{\infty}(\G))},\vphi_F|_{\mc{Z}(\ell^{\infty}(\G))})-S(\mu,\mu\circ E).\end{equation*}
However, since the density of $\mu$ is central, $\mu\circ E=\mu$, implying $S(\mu,\mu\circ E)=0$ (cf. \cite[pg. 16]{OP}). Since $\mc{Z}(\ell^{\infty}(\G))$ is commutative, we obtain
\begin{equation*}H(\mu,\vphi)=-\sum_{\alpha\in F}\mu_\alpha\log\bigg(\frac{\mu_\alpha}{\tr(F^\alpha)^2}\bigg).\end{equation*}

Next, let $\mu$ be arbitrary, and put $\xi:=\Lm_\vphi(D_\mu^{1/2})$, which lies in $\mc{D}(\log(d\mu/d\vphi'))$ by hypothesis. For a finite subset $F\subseteq I$, write $\mu_F:=\sum_{\alpha\in F}\mu_\alpha\delta_\alpha$ and $\mu_{F^c}:=\sum_{\alpha\in F^c}\mu_\alpha\delta_\alpha$. Note that $\mu_F$ and $\mu_{F^c}$ have orthogonal support, so by properties of the spatial derivative (cf. \cite[Corollary 12]{C}) we have
\begin{equation*}\frac{d\mu}{d\vphi'}=\frac{d\mu_F}{d\vphi'}\oplus\frac{d\mu_{F^c}}{d\vphi'}.\end{equation*}
By functional calculus, the above decomposition also holds for the corresponding logarithms, i.e., $\log\bigg(\frac{d\mu}{d\vphi'}\bigg)=\log\bigg(\frac{d\mu_F}{d\vphi'}\bigg)\oplus\log\bigg(\frac{d\mu_{F^c}}{d\vphi'}\bigg)$. Therefore, $z_F\log\bigg(\frac{d\mu}{d\vphi'}\bigg)\subseteq\log\bigg(\frac{d\mu}{d\vphi'}\bigg)z_F$, so that
\begin{equation*}\bigg\la\log\bigg(\frac{d\mu_F}{d\vphi'}\bigg)\xi,\xi\bigg\ra=\bigg\la\log\bigg(\frac{d\mu_F}{d\vphi'}\bigg)z_F\xi,\xi\bigg\ra=\bigg\la\log\bigg(\frac{d\mu}{d\vphi'}\bigg)z_F\xi,\xi\bigg\ra=\bigg\la\log\bigg(\frac{d\mu}{d\vphi'}\bigg)\xi,z_F\xi\bigg\ra.\end{equation*}
Putting things together,
\begin{equation*}H(\mu,\vphi)=-\lim_{F}\bigg\la\log\bigg(\frac{d\mu_F}{d\vphi'}\bigg)\xi,\xi\bigg\ra=-\lim_{F}\sum_{\alpha\in F}\mu_\alpha\log\bigg(\frac{\mu_\alpha}{\tr(F^\alpha)^2}\bigg)=H_\sigma(\mu).\end{equation*}
\end{proof}

We therefore see that restricted to $q$-traces, our entropy coincides with that of Hiai--Izumi, suggesting that this is the appropriate concept to
extend the entropy theory of random walks on discrete quantum groups beyond $q$-traces. A natural question is then how this entropy behaves under
quantum group convolution. Furthermore, to what extent does the uncertainty principle generalize to non-unimodular quantum groups? These questions will be pursued elsewhere.

\section*{Acknowledgements}

The authors would like to thank the referee for many valuable comments, especially concerning a simpler proof of Theorem \ref{last}.

\end{spacing}


\vspace{0.2in}

\begin{thebibliography}{00}

\bibitem{AR} G. Alagic \and A. Russell,
\e{Uncertainty principles for compact groups},
Illinois J. Math. {\bf 52} (2008), no. 4, 1315--1324.

\bibitem{Bec} W. Beckner,
\e{Inequalities in Fourier analysis},
Ann. of Math. {\bf 2} 102 (1975), no. 1, 159--182.

\bibitem{Choda} M. Choda,
\e{Von Neumann entropy and relative position between subalgebras},
Internat. J. Math. {\bf 24} (2013), no. 8, 1350066, 19 pp.

\bibitem{C} A. Connes,
\e{On the spatial theory of von Neumann algebras},
J. Funct. Anal. {\bf 35} (1980), 153--164.

\bibitem{CS} A. Connes \and E. St\o rmer,
\e{Entropy for automorphisms of $II_1$ von Neumann algebras},
Acta Math. {\bf 134} (1975), no. 3-4, 289--306.

\bibitem{Coo} T. Cooney,
\e{A Hausdorff--Young inequality for locally compact quantum groups},
Internat. J. Math. {\bf 21} (2010), no. 12, 1619--1632.

\bibitem{ES} M. Enock \and J. M Schwartz,
\e{Kac Algebras and Duality of Locally Compact Groups},
Springer--Verlag, Berlin Heidelberg 1992.

\bibitem{FK} T. Fack \and H. Kosaki,
\e{Generalized $s$-numbers of $\tau$-measurable operators},
Pacific J. Math. {\bf 123} (1986), no. 2, 269--300.

\bibitem{FL} R. L. Frank and E. H. Lieb,
\e{Entropy and the uncertainty principle},
Ann. Henri Poincar\'{e} {\bf 13} (2012), no. 8, 1711--1717.

\bibitem{Haa} U. Haagerup,
\e{Decomposition of completely bounded maps on operator algebras},
Unpublished manuscript (1980).

\bibitem{Hei} W. Heisenberg,
\e{\"{U}ber den anschaulichen Inhalt der quantummechanischen Kinematik und Mechanik},
Zeitschrift f\"{u}r Physik. {\bf 43} (1927), pp. 172--198.

\bibitem {HI} F. Hiai \and M. Izumi,
\textit{Amenability and strong amenability for fusion algebras with applications to subfactor theory},
 Internat. J. Math. {\bf 9} (1998), no. 6, 669--722.

\bibitem{H} M. Hilsum,
\e{Les espaces $L_p$ d'une alg\`{e}bre de von Neumann d\'{e}finies par la deriv\'{e}e spatiale},
J. Funct. Anal. {\bf 40} (1981), no. 2, 151--169.

\bibitem{Hir}
I. I. Hirschman Jr.,
\e{A note on entropy},
Amer. J. Math. {\bf 79} (1957),  152--156.

\bibitem {I} M. Izumi,
\textit{Non-commutative Poisson boundaries and compact quantum group actions},
Adv. Math. \textbf{169} (2002), no. 1, 1--57.


\bibitem {JNR} M. Junge, M. Neufang \and Z.-J. Ruan,
\textit{A representation theorem for locally compact quantum
groups}, Internat. J. Math. \textbf{20} (2009), 377--400.

\bibitem {Kai-Ver} V. A. Kaimanovich \and A. M. Vershik,
{\it Random walks on discrete groups: boundary and entropy}, Ann.  Probab.
\textbf{11} (1983), 457--490.

\bibitem{K} E. H. Kennard,
\e{Zur Quantenmechanik einfacher Bewegungstypen},
Zeitschrift f\"{u}r Physik. {\bf 44} (1927), pp. 1--25.

\bibitem {KV} J. Kustermans \and S. Vaes, \textit{Locally compact quantum groups},
Ann. Sci. Ecole Norm. Sup. \textbf{33} (2000), 837--934.


\bibitem{OP} M. Ohya \and D. Petz,
\e{Quantum Entropy and Its Use},
Texts and Monographs in Physics, Springer--Verlag, Berlin Heidelberg 1993.

\bibitem{P} K. R. Parthasarathy,
\e{On estimating the state of a finite level quantum system},
Infin. Dimens. Anal. Quantum Probab. Relat. Top. {\bf 7} (2004), no. 4, 607--617.


\bibitem{Petz} D. Petz,
\e{Complementarity in quantum systems},
Rep. Math. Phys. {\bf 59} (2007), no. 2, 209--224.

\bibitem{PSW} D. Petz, A. Sz\`{a}nt\'{o} \and M. Weiner,
\e{Complementarity and the algebraic structure of four-level quantum systems},
Infin. Dimens. Anal. Quantum Probab. Relat. Top. {\bf 12} (2009), no. 1, 99--116.

\bibitem {Popa83} S. Popa,
\textit{Orthogonal pairs of $*$-subalgebras in finite von Neumann algebras},
J. Operator Theory {\bf 9} (1983), no. 2, 253--268.

\bibitem{Ruskai} M. B. Ruskai,
\e{Inequalities for traces on von Neumann algebras},
Comm. Math. Phys. {\bf 26} (1972), 280--289.


%\bibitem{SS} C. K. Seng \and N. W. Seng,
%\e{A simple proof of the uncertainty principle for compact groups},
%Expos. Math. {\bf 23} (2005), 147–-150.


\bibitem {Tak2} M. Takesaki,
Theory of Operator Algebras. II. Encyclopaedia of Mathematical Sciences 125.
Operator Algebras and Non-commutative Geometry, 6. Springer-Verlag, Berlin (2003)

\bibitem{T} M. Terp,
\e{$L_p$ Fourier transformation on non-unimodular locally compact groups},
preprint 1980.

\bibitem{T2} M. Terp,
\e{Interpolation spaces between a von Neumann algebra and its predual},
J. Operator Theory {\bf 8} (1982), no. 2, 327--360.

\bibitem{W} S. L. Woronowicz,
\e{Compact matrix pseudogroups},
Comm. Math. Phys. {\bf 111} (1987), no. 4, 613--665.

\bibitem{VV} S. Vaes \and A. Van Daele,
\e{The Heisenberg commutation relations, commuting squares and the Haar measure on locally compact quantum groups}, in Operator Algebras and Mathematical Physics (Constanta, 2001) (Theta, Bucharest, 2003), pp. 379--400.


\end{thebibliography}
\end{document}